\documentclass[letterpaper, 10 pt, conference]{ieeeconf}  

\IEEEoverridecommandlockouts                              

\usepackage{amsmath,amssymb,amsfonts,amsthm, url, algorithm,algpseudocode}
\usepackage{dsfont}
\usepackage{xfrac}
\usepackage{cite}
\usepackage{graphicx}

\newtheorem{theorem}{Theorem}[]

\newtheorem{proposition}[theorem]{Proposition}
\newtheorem{definition}[]{Definition}
\newtheorem{assumption}[]{Assumption}
\newtheorem{remark}[]{Remark}
\newcommand{\N}{\mathbb{N}}
\newcommand{\Z}{\mathbb{Z}}
\newcommand{\R}{\mathbb{R}}
\newcommand{\Prob}{\mathbb{P}}
\newcommand{\hquad}{\hspace{0.5em}} 
\newcommand{\indep}{\perp \!\!\! \perp}
\title{\LARGE \bf Moment State Dynamical Systems for Nonlinear Chance-Constrained Motion Planning}
\author{Allen Wang, Ashkan Jasour, Brian Williams \thanks{All authors are with the Computer Science and Artificial Intelligence Laboratory (CSAIL), Massachusetts Institute of Technology
        {\tt\small \{allenw, jasour, williams @ mit.edu\}}}}
\begin{document}
\maketitle
\thispagestyle{empty}
\pagestyle{empty}
\begin{abstract}
Chance-constrained motion planning requires uncertainty in dynamics to be propagated into uncertainty in state. When nonlinear models are used, Gaussian assumptions on the state distribution do not necessarily apply since almost all random variables propagated through nonlinear dynamics results in non-Gaussian state distributions. To address this, recent works have developed moment-based approaches for enforcing chance-constraints on non-Gaussian state distributions. However, there still lacks fast and accurate moment propagation methods to determine the necessary statistical moments of these state distributions. To address this gap, we present a framework that, given a stochastic dynamical system, can algorithmically search for a new dynamical system in terms of \textit{moment state} that can be used to propagate moments of disturbance random variables into moments of the state distribution. The key algorithm, \textit{TreeRing}, can be applied to a large class of nonlinear systems which we refer to as \textit{trigonometric polynomial systems}. As an example application, we present a distributionally robust RRT (DR-RRT) algorithm that propagates uncertainty through the nonlinear Dubin's car model without linearization.
\end{abstract}
\section{Introduction}
Motivated by the inherently uncertain nature of real-world systems, significant amounts of recent research has aimed at developing motion planning algorithms that explicitly reason about uncertainty in system dynamics, often with uncertainty modeled probabilistically. When probabilistic uncertainty is introduced, \textit{chance-constraints} are a natural way to encode a desired level of safety in the motion planning problem. Historically, algorithms for chance-constrained motion planning under dynamics uncertainty have been restricted to the case of linear dynamics and additive Gaussian uncertainty \cite{CCRRT,LQG1,Larss2, aoude2013probabilistically}. This assumption is popular in part because linear transformations of multivariate Gaussians are still multivariate Gaussian, thus allowing for cumulative density functions to be easily encoded. However, such representations can be limiting and inaccurate as most real world systems are nonlinear and many sources of uncertainty are non-Gaussian. For one, almost all distributions propagated through nonlinear dynamics will produce a non-Gaussian state distribution. To handle nonlinear systems and non-Gaussian uncertainty, some prior works adopt Monte Carlo style approaches and sample the sources of uncertainty to fix the random variables at each of the sample values, generating multiple deterministic dynamical systems. These multiple systems can then be used either as a large number of constraints in an optimization routine or as a sub-routine in a sampling-based motion planner \cite{calafiore2006scenario, Larss1, melchior2007particle}. Such approaches, however, usually become computationally intractable as a relatively large number of samples is needed to certify chance-constraints.

To avoid the computational intractability of sampling, many recent works adopt approaches that can establish upper bounds on state constraint violation using statistical moments of random variables. For example, higher order statistical moments of a random variable can be used in a sums-of-squares (SOS) program to find upper bounds on risk that are optimal w.r.t. the supplied moment information \cite{jasour2018moment, jasour2019risk,bertsimas2005optimal}. When only low order moments, such as mean and covariance, are used, concentration inequalities can be used to establish simple analytic upper bounds on the probability of state constraints being violated \cite{wang2020non,hewing2019cautious}.

To enforce chance-constraints on state, moments of the state distribution need to be known, but it is challenging to determine moments of state when random variables are propagated through nonlinear dynamics. Due to this difficulty, most prior works that utilize moment-based approaches resort to linearizing the dynamics to perform moment propagation \cite{calafiore2006distributionally, dr_rrt,hewing2019cautious,carron2019data}. To address the problem of moment propagation for Gaussian driven stochastic processes, polynomial chaos expansions (PCE) are a well-studied method in the uncertainty quantification community to determine statistics of random variables propagated through some nonlinear function \cite{xiu2010numerical}. A few efforts have been made towards using PCE in control and robotics applications \cite{fisher2011optimal, nakka2019trajectory,kim2012generalized}, but its practicality thus far has been limited by: 1) lack of guarantees on approximation error 2) the large number of terms that are often needed 3) the necessity of calling expensive quadrature routines to determine expansion coefficients.

In this paper, we address the problem of moment propagation through nonlinear systems by developing a framework that can search for a new dynamical system in terms of \textit{moment state}. We present this framework in Section \ref{sec:poly_moment_pro} and show that it can produce a polynomial system or even a \textit{linear system} at the cost of increasing dimensionality. Given the statistical moments of the disturbance variables, these systems are in closed form and, thus, can be evaluated rapidly. In Section \ref{sec:trig_poly_systems}, we extend the range of applicability of this framework to \textit{trigonometric polynomial} systems; systems that have sines or cosines in the system state variables but can be transformed into polynomial systems by a change of variables. This is significant as it extends our method to encompass a wide range of systems important in control and robotics such as Dubin's car. Finally, in Section \ref{sec:dr_rrt}, we provide an example application of this moment propagation framework to develop an improved distributionally robust RRT (DR-RRT) that uses a stochastic Dubin's car model without relying on linearization as was previously required. In Section \ref{sec:comparison}, we provide examples of the accuracy and run-times of our approach compared to linearization and naive Monte Carlo. The source code for \textit{TreeRing} is available at \url{github.com/wangmengyu96/TreeRing}.

\section{Moment Propagation for Polynomial Systems}\label{sec:poly_moment_pro}
This section presents the methodology for finding a moment state dynamical system. Throughout this section, we will be concerned with the system:
\begin{align}
    \mathbf{x}_{t+1} = f(\mathbf{x}_t, \mathbf{w}_t)\label{eq:general_poly_system}
\end{align}
Where $\mathbf{x}_t$ is the $n_x$ dimensional state random vector, $\mathbf{w}_t$ is a $n_w$ dimensional disturbance random vector that is assumed to be independent of $\mathbf{x}_t$, and $f$ is a vector-valued function that is polynomial in $\mathbf{x}_t$ and $\mathbf{w}_t$. This formalism implicitly models control inputs as deterministic control variables can be represented in $\mathbf{w}_t$ as random variables that take on a value with probability one. We assume that the moments of $\mathbf{w}_t$ are known for each time step; we argue this is not a particularly restrictive requirement as moments can be computed via auto-differentiation of the moment generating function of known distributions. While we restrict our focus to discrete time systems in this paper, we would like to note that there are standard discrete time approximations of continuous time stochastic differential equations, such as the Euler-Maruyama and stochastic Runge-Kutta methods, that are analogous to the well known ones for deterministic systems \cite{kloeden2013numerical}.

\ref{subsec:multi_idx} introduces multi-index notation which will be used frequently throughout this paper. \ref{subsec:poly_properties} then establishes several properties of polynomial systems relevant to moment propagation. In \ref{subsec:moment_update_forms} and \ref{subsec:moment_prop_basis}, we show that a moment state dynamical system can be found if we find a \textit{complete moment basis} for the set of moments we are interested in. Finally, \ref{subsec:tree_ring} presents an algorithm for finding complete moment bases, which we call \textit{TreeRing}.

\subsection{Multi-Index Notation}\label{subsec:multi_idx}
For simplicity, we adopt multi-index notation. That is, for any $n_x$ dimensional vector $\mathbf{x}_t = [x_{1_t}, ..., x_{{n_x}_t}]^T$ and multi-index $\alpha\in\N^{n_x}$, we adopt the following short-hand:
\begin{align}
    \mathbf{x}_t^\alpha = \prod_{i=1}^{n_x} x_{i_t}^{\alpha_i}\label{eq:multi_index}
\end{align}
All multi-indices are denoted with lower case Greek letters. This provides much simpler expressions for polynomials and for moments of random vectors. For example, if we want to express the moment $\mathbb{E}[x_{1_t}^3x_{3_t}]$ in multi-index notation, we can simply let $\alpha = (3, 0, 1, 0, ..., 0)$ and write $\mathbb{E}[\mathbf{x}_t^\alpha]$. One of the advantages of this formalism is the existence of a one-to-one mapping between moments of $\mathbf{x}_t$ and multi-indices in $\mathbb{N}^{n_x}$, and similarly for $\mathbf{w}_t$, so it is equivalent to talk about one or the other.

\subsection{Properties of Polynomial Systems}\label{subsec:poly_properties}
We begin by reviewing a few elementary properties of polynomials arising from their ring structure that are relevant to Proposition \ref{prop:vec_moments_poly}. In Proposition \ref{prop:closure} and \ref{prop:degrees}, $K[\mathbf{v}]$ denotes the ring of polynomials in the vector $\mathbf{v}$ over the field $K$. In this paper, we will only work with $\R[\mathbf{v}]$ and $\mathbb{C}[\mathbf{v}]$, but we adopt the more general notation for concision in the propositions. Proposition \ref{prop:closure} states a few closure properties; these are of importance because simulating the discrete-time system essentially consists of applying these operations.
\begin{proposition}\label{prop:closure}
(Closure Properties) For any $p_1, p_2\in K[\mathbf{v}]:$
\leavevmode
\begin{enumerate}
    \item $p_1 + p_2\in K[\mathbf{v}]$ (Addition) 
    \item $p_1p_2\in K[\mathbf{v}]$  (Multiplication)
    \item $p_1^n\in K[\mathbf{v}]$ (Exponentiation)
    \item $(p_1\circ p_2)\in K[\mathbf{v}]$ (Function Composition)
\end{enumerate}
\end{proposition}
A useful heuristic for evaluating the complexity of a polynomial expression is its degree. In general, it's undesirable for the degree of polynomials to increase as the application of operations to high degree polynomials is likely to result in a polynomial with a large number of terms. Proposition \ref{prop:degrees} states the degree of the resulting polynomial after applying the closure operations.
\begin{proposition}\label{prop:degrees}
For any $p_1, p_2 \in K[\mathbf{v}]:$
\leavevmode
\begin{enumerate}
    \item $\text{deg}(p_1 + p_2) = \max\{\text{deg}(p_1), \text{deg}(p_2)\}$
    \item $\text{deg}(p_1p_2) = \text{deg}(p_1) + \text{deg}(p_2)$
    \item $\text{deg}(p_1^n) = \text{deg}(p_1)^n$ where $n\in\N$
    \item $\text{deg}(p_1\circ p_2) = \text{deg}(p_1)\text{deg}(p_2)$
\end{enumerate}
\end{proposition}
In the case of vector-valued polynomials, the degree operator returns the sum of the vector of degrees of each entry of the vector-valued polynomial which we will refer to as the \textit{degree vector}. To compute its degree after applying an operation, simply apply the above proposition to each of the scalar-valued entries and sum the resulting values.

Closure under exponentiation is a particularly important property as it implies that raising any polynomial $p(\mathbf{v})$ in a random vector $\mathbf{v}$ to the power of $n\in\N$ results in a polynomial $p(\mathbf{v})^n$. By applying the linearity of expectation to $\mathbb{E}[p(\mathbf{v})^n]$, we see that it can be computed in terms of moments of $\mathbf{v}$. By simply extending this idea to the case where $p$ is a vector-valued function and $p(\mathbf{v})$ is raised to a multi-index, we arrive at Proposition \ref{prop:vec_moments_poly}a which is a simple extension of Proposition 1 in \cite{wang2020non} from the scalar-valued to the vector-valued case. Proposition \ref{prop:vec_moments_poly}b then states the degree of the resulting polynomial.

\begin{proposition}\label{prop:vec_moments_poly}
\leavevmode
a) For any $n$ dimensional random vector $\mathbf{v}$, any vector-valued polynomial $p\in\R[\mathbf{v}]$ s.t. $p:\R^n\rightarrow\R^m$ with degree vector $\sigma$ and any $\alpha\in\N^m$, there exists a set of multi-indices, $\mathcal{B}[\alpha]\subset\mathbb{N}^n$, and coefficients $\{c_\beta\in\R : \beta\in\mathcal{B}[\alpha]\}$ s.t:
\begin{align}
    \mathbb{E}[p(\mathbf{v})^\alpha] = \sum_{\beta\in\mathcal{B}[\alpha]}c_\beta\mathbb{E}[\mathbf{v}^\beta]
\end{align}
b) $\text{deg}\left(p(\mathbf{v})^\alpha\right) =  \sum_{i=1}^n \alpha_i\sigma_i$
\end{proposition}
\begin{proof}
a) See \cite{wang2020non}. b)
$p_i$ has degree $\sigma_i$, so $p_i(\mathbf{v})^{\alpha_i}$ has degree $\sigma_i\alpha_i$ in $\mathbf{v}$ by Proposition \ref{prop:degrees}.3. By Proposition \ref{prop:degrees}.2, $p(\mathbf{v})^\alpha$ thus has degree $\sum_{i=1}^n\alpha_i\sigma_i$.
\end{proof}
\begin{remark}
While $p(\mathbf{v})$ is a vector-valued polynomial in $\mathbf{v}$, $p(\mathbf{v})^\alpha$ is a scalar-valued polynomial in Proposition \ref{prop:vec_moments_poly}.
\end{remark}
\subsection{Moment Update Forms}\label{subsec:moment_update_forms}

If we treat the state and disturbance vectors as one vector in our polynomial system, we can apply Proposition \ref{prop:vec_moments_poly}a and see that any moment of $\mathbf{x}_{t+1}$ can be computed in terms of moments of the joint distribution of $\mathbf{x}_t$ and $\mathbf{w}_t$. Since we assumed that the disturbance is independent of state, we have that moments of $\mathbf{x}_{t+1}$ can be expressed in the form of equation \ref{eq:exog_state_moments}, the moment update form, by splitting the multi-index s.t. $\beta = [\beta_x, \beta_w]$. Note that by Proposition \ref{prop:vec_moments_poly}b, the polynomial in the moment update form is of relatively low order: $\sum_{i=1}^{n_x} \alpha_i\sigma_i$ to be precise.
\begin{definition}
Given the polynomial system, equation \ref{eq:general_poly_system}, and some multi-index $\alpha\in\mathbb{N}^{n_x}$ corresponding to a moment of the state vector, the \textbf{moment update form} (MUF) is:
\begin{align}
\mathbb{E}[\mathbf{x}_{t + 1}^\alpha] = \sum_{\beta\in\mathcal{B}[\alpha]}c_\beta \mathbb{E}\left[\mathbf{x}_t^{\beta_x}\right]\mathbb{E}\left[\mathbf{w}_t^{\beta_w}\right]\label{eq:exog_state_moments}
\end{align}
\end{definition}
Throughout this paper, we will denote the multi-index set of the moment update form induced by $\alpha$ as $\mathcal{B}[\alpha]$ and denote the sets of the state and disturbance components respectively by $\mathcal{B}_x[\alpha]$ and $\mathcal{B}_w[\alpha]$.
\begin{remark}In practice, moment update forms can be easily found with standard symbolic algebra packages. In this paper, we used the the built-in functionality in SymPy to express symbolic expressions as polynomials in a set of variables \cite{meurer2017sympy}.
\end{remark}
The moment update form essentially states that moments of the state vector at time $t + 1$ are \textit{linear} in the moments of the state vector at time $t$ when given moments of external disturbances.  Now if there is independence between elements of the state vector, $\mathbb{E}[\mathbf{x}_t^{\beta_x}]$ may be factored into the product of lower order moments to arrive at what we will refer to as the \textit{reduced moment update form}. There are ultimately interesting tradeoffs between using the reduced and non-reduced moment update forms that we will discuss in the following subsection. 
\begin{definition} \label{def:reduced_moment_update}
Given the polynomial system, equation \ref{eq:general_poly_system}, and some multi-index $\alpha\in\mathbb{N}^{n_x}$ corresponding to a moment of the state vector, the \textbf{reduced moment update form} is: 
\begin{align} 
    \mathbb{E}[\mathbf{x}_{t+1}^\alpha] = \sum_{\beta\in\mathcal{B}[\alpha]}c_\beta\mathbb{E}\left[\mathbf{w}_t^{\beta_w}\right] \prod_{i=1}^{n_\beta}\mathbb{E}\left[\mathbf{x}_t^{\beta_x^{(i)}}\right]\label{eq:reduced_moment}
\end{align} 
Where $\sum_{i=1}^{n_\beta}\beta_x^{(i)} = \beta_x$, the expression involving the product operator equals $\mathbb{E}[\mathbf{x}_t^{\beta_x}]$ and for each $\beta_x^{(i)}$, $\mathbb{E}\left[\mathbf{x}_t^{\beta_x^{(i)}}\right]$ can not be factored further.
\end{definition}
The reduced moment update form requires the additional step of factorizing $\mathbb{E}[\mathbf{x}_t^{\beta_x}]$, but this can be performed with standard search algorithms to find connected components on the, manually specified, \textit{dependence graph} of $\mathbf{x}_t$.

\begin{definition}\label{def:dependence_graph}
The graph $G_{\mathbf{x}} = (V_\mathbf{x}, E_\mathbf{x})$ is the \textbf{dependence graph} of $\mathbf{x}$ where $V_\mathbf{x}$ is the set of variables in $\mathbf{x}$ and $E_{\mathbf{x}}$ is the set of undirected edges s.t. $(x_i, x_j)\in E_{\mathbf{x}}$ i.f.f. $x_i$ and $x_j$ are dependent.
\end{definition}
However, pairwise independence of a collection of random variables does not necessarily imply mutual independence. This is a difficult issue to resolve in general, so we make the following assumption on pairwise independence of state variables.
\begin{assumption}
Suppose $\mathbf{a}$ and $\mathbf{b}$ are random vectors that are subsets of $\mathbf{x}_t$. Then pairwise independence between elements of $\mathbf{a}$ and $\mathbf{b}$ implies $\mathbf{a}$ and $\mathbf{b}$ are independent. That is, letting $f_{\mathbf{a}}, f_{\mathbf{b}}, f_{\mathbf{a},\mathbf{b}}$ denote probability density functions:
\begin{align}
    f_{\mathbf{a},\mathbf{b}}(\cdot, \cdot) = f_\mathbf{a}(\cdot)f_\mathbf{b}(\cdot)\Leftrightarrow \forall a\in\mathbf{a},\forall b\in\mathbf{b}, a\indep b\label{eq:factor_example}
\end{align}
\end{assumption}
Note that moments of joint distributions can be factored in the same way as probability density functions above. With these facts in mind, the goal now is to find a partition of $\mathbf{x}_t$ into multiple random vectors s.t. equation \ref{eq:factor_example} can be applied. It should be easy to see that the set of connected components of the dependence graph forms the correct partition because by definition, pairwise independence corresponds to the existence of an edge. With the example in Fig. \ref{fig:dependence_graph}, we can perform the following factorization:
\begin{align}
    f_{\mathbf{x}}(\cdot) = f_{x_1, x_2, x_3, x_4}(\cdot)f_{x_5, x_6}(\cdot)f_{x_7}(\cdot)f_{x_8}(\cdot)
\end{align}
Thus, for any given multi-index $\alpha$, we can find the sub-graph of the dependence graph containing the variables with non-zero entries in $\alpha$ and then find a factorization by finding the connected components of the resulting sub-graph.
\begin{figure}
    \vspace*{0.3cm}
    \centering
    \includegraphics[width = 0.7\linewidth]{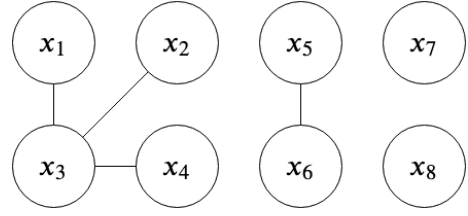}
    \caption{Example of a Dependence Graph.}
    \label{fig:dependence_graph}
\end{figure}

\subsection{Moment Bases}\label{subsec:moment_prop_basis}
Since we assume the moments of the external disturbances are known, the (possibly reduced) moment update form seems to suggest there may be a way to recursively determine the desired moments of the state vector across a finite time horizon when given some initial condition. In fact, this is possible if we find a \textit{complete moment basis}. Suppose we are interested in determining a set of moments of the state vector $\mathbf{x}_{t+1}$; we will denote the set of multi-indices corresponding to these moments by $\mathcal{A}\subset\N^{n_x}$. Such sets of multi-indices will be referred to as \textit{moment bases}. Given a moment basis, we denote the \textit{moment state} of the system by:
\begin{align}
    \mathbf{x}_t[\mathcal{A}] := \{\mathbb{E}[\mathbf{x}_t^\alpha] : \alpha\in\mathcal{A}\}
\end{align}
This notation also lends itself to expressing the set of all moments of $\mathbf{w}_t$ that are needed. Recall that for any given $\alpha$, $\mathcal{B}_w[\alpha]$ denotes the set of disturbance multi-indices in the (possibly reduced) MUF, so for each $\alpha\in\mathcal{A}$ we can denote the set of disturbances needed for that particular $\alpha$ by:
\begin{align}
    \mathbf{w}_t\left[\mathcal{B}_w[\alpha]\right] := \{\mathbb{E}[\mathbf{x}_t^{\beta_w}] : \beta_w\in\mathcal{B}_w[\alpha]\}
\end{align}
By taking the union across $\mathcal{A}$, we arrive at the set of disturbance moments needed. For brevity, we will denote:
\begin{align}
    \mathbf{w}_t\left[\mathcal{B}_w[\mathcal{A}]\right] := \bigcup_{\alpha\in\mathcal{A}}\mathbf{w}_t[\mathcal{B}_w[\alpha]]
\end{align}

We are now ready for the definitions of complete moment bases w.r.t. the (possibly reduced) MUF. When a moment basis is not complete, we will want to find a \textit{completion}.
\begin{definition}\label{def:complete_basis}
$\mathcal{A}$ is a \textbf{complete moment basis w.r.t. the moment update form} if $\forall \alpha\in\mathcal{A}$, the multi-index set in the moment update form corresponding to the state vector, $\mathcal{B}_x[\alpha]$, is a subset of $\mathcal{A}$.
\end{definition}
\begin{definition}\label{def:complete_basis_reduced}
$\mathcal{A}$ is a \textbf{complete moment basis w.r.t. the reduced moment update form} if $\forall \alpha\in\mathcal{A}$,  $\forall\beta\in\mathcal{B}_x[\alpha]$, $\beta_x^{(i)}\in\mathcal{A}$.
\end{definition}
\begin{definition}\label{def:completion}
A \textbf{completion of a moment basis} denoted $\bar{\mathcal{A}}$ is a super-set of $\mathcal{A}$ that is a complete moment basis w.r.t. the (possibly reduced) moment update form.
\end{definition}
In both cases, completeness of the moment bases is a sufficient condition that the (possibly reduced) MUFs of the elements in the moment basis form a moment state dynamical system. This should be clear as completeness of $\mathcal{A}$ implies that every moment in $\mathbf{x}_{t+1}[\mathcal{A}]$ can be expressed explicitly in terms of $\mathbf{x}_t[\mathcal{A}]$ and moments of $\mathbf{w}_t$.

\begin{definition}
A \textbf{moment state dynamical system} given a complete moment basis $\mathcal{A}$ is:
\begin{align}
    \mathbf{x}_{t+1}[\mathcal{A}] = h(\mathbf{x}_t[\mathcal{A}], \mathbf{w}_t\left[\mathcal{B}[\mathcal{A}]\right])
\end{align}
Where $h$ is a vector-valued function formed by the (possibly reduced) moment update forms of $\alpha\in\mathcal{A}$.
\end{definition}

Ultimately, both Definition \ref{def:complete_basis} and \ref{def:complete_basis_reduced} simply try to capture the idea that every moment at time step $t+1$ can be computed in terms of moments in the moment state at the previous time step. However, the distinction is important because a complete moment basis w.r.t. the un-reduced MUF has \textit{linear time-varying dynamics}. This can be easily seen by inspection of equation \ref{eq:exog_state_moments} because moments of $\mathbf{w}_t$ are assumed to be known and thus can be treated as constants. However, empirically, we find that using the reduced MUF has the advantage of producing smaller completions. It also makes intuitive sense that using reduced MUFs will result in smaller completions as the reduced MUF tries to express moments in terms of the lowest order moments possible.


\subsection{TreeRing: An Algorithm for Finding Completions of Moment Bases}\label{subsec:tree_ring}
If a moment basis $\mathcal{A}$ is incomplete, a simple step to take would be to simply expand the moment basis by adding the multi-indices that are not in $\mathcal{A}$ to $\mathcal{A}$ and determining their MUFs. Recursively doing so produces an algorithm analogous to tree search, where the nodes are moments of the state vector that show up in MUFs, and the nodes are expanded if the moment is not in the current moment basis. We call this algorithm \textit{TreeRing} as it is essentially tree search, but over the ring of polynomials. TreeRing in Algorithm \ref{alg:treering} uses the reduced MUF, but an algorithm using the un-reduced MUF can be arrived at by simply removing line number \ref{eq:reduced_moment_line}, letting $n_\beta = 1$, and replacing $\beta_x^{(i)}$ with $\beta$ everywhere. If the algorithm terminates, then we have that the resulting $\mathcal{A}$ is a complete moment basis. By finding all of the (possibly reduced) MUFs associated with moments in the complete moment basis, we arrive at our moment state dynamical system.
\begin{algorithm}
\caption{TreeRing}\label{alg:treering}
\begin{algorithmic}[1]
\Procedure{Expand}{$\alpha$, $\mathbf{x}_t$, $\mathbf{w}_t$, $f$, $\mathcal{A}$}
    \State $\mathcal{B}[\alpha], C_{\mathcal{B}[\alpha]}\leftarrow$ Express $\mathbb{E}[f(\mathbf{x}_{t+1}, \mathbf{w}_{t+1})^\alpha]$ in MUF
    \State Add $\alpha$ to $\mathcal{A}$
    \ForAll{$\beta\in\mathcal{B}[\alpha]$}
        \State Factor $\mathbb{E}[\mathbf{x}_t^{\beta_x}]$ to generate $n_\beta$ multi-indicies $\beta_x^{(i)}$\label{eq:reduced_moment_line}
        \ForAll{$i\in[n_\beta]$}
            \If{$\beta_x^{(i)}\notin\mathcal{A}$}
                \State \text{Expand}($\beta_x^{(i)}, \mathbf{x}_t, \mathbf{w}_t, f, \mathcal{A}$)
            \EndIf
        \EndFor
    \EndFor
\EndProcedure
\end{algorithmic}
\end{algorithm}
\section{Extension to Trigonometric Polynomial Systems}\label{sec:trig_poly_systems}
Many dynamical systems in control and robotics are not polynomial, but can be made polynomial by a change of variables replacing sines and cosines with new indeterminants. We will refer to such systems as \textit{trigonometric polynomial systems}. For example, consider the following Dubin's car model with uncertainty in actuation (note that the control inputs are absorbed as constant shifts in $w_{v_t}$ and $w_{\theta_t}$ for brevity):
\begin{subequations}
\begin{align}
    x_{t+1} &= x_t + v_t\cos(\theta_t)\\
    y_{t+1} &= y_t + v_t\sin(\theta_t)\\
    v_{t+1} &= v_t + w_{v_t}\\
    \theta_{t+1} &= \theta_t + w_{\theta_t}
\end{align}
\end{subequations}
We can transform the above system into a polynomial system by making the substitutions $c_t = \cos(\theta_t)$ and $s_t = \sin(\theta_t)$.
\begin{subequations}\label{eq:dubins_poly}
\begin{align}
    x_{t + 1} &= x_t + v_tc_t\\
    y_{t + 1} &= y_t + v_ts_t\\
    v_{t + 1} &= v_t + w_{v_t}\\
    c_{t + 1} &= c_tc_{w_t} - s_ts_{w_t}\\
    s_{t + 1} &= s_tc_{w_t} + c_ts_{w_t}
\end{align}
\end{subequations}
Update relations for $c_t$ and $s_t$ were arrived at by using the trigonometric sums formulas to expand out $\cos(\theta_t + w_{\theta_t})$ and $\sin(\theta_t + w_{\theta_t})$ into the expressions shown in equation \ref{eq:dubins_poly} where $c_{w_{t}} = \cos(w_{\theta_t})$ and $s_{w_t} = \sin(w_{\theta_t})$.
To apply the methods developed for polynomial systems, however, we need to be able to compute the moments of $c_{w_t}$ and $s_{w_t}$. In \ref{subsec:trig_moments}, we show that these \textit{trigonometric moments} can be computed directly in terms of the characteristic functions of the original random variables. This retains the generality of our results on polynomial systems as extensive catalogues of characteristic functions for distributions are available.

\subsection{Computing Trigonometric Moments}\label{subsec:trig_moments}
In this subsection, we show moments of the form:
\begin{align}
    \mathbb{E}[\cos^m(X)\sin^n(X)]
\end{align}
for any random variable $X$, can be computed in terms of the characteristic function of $X$, defined as:
\begin{align}
    \Phi_X(t) = \mathbb{E}[e^{itX}]
\end{align}
The characteristic function is simply the Fourier transform of the PDF of $X$, so it makes intuitive sense that there exists deep connections between it and trigonometric functions. We provide explicit closed form solutions for the pure moment case, i.e when one of $m$ or $n$ equals zero, and provide an algorithmic approach for finding an expression using symbolic algebra in the case when both $m, n>0$. The first major insight is that moments of $\cos(X)$ and $\sin(X)$ may be related to the characteristic function via Euler's identity:
\begin{subequations}
\begin{align}
    \Phi_{X}(t) &= \mathbb{E}[e^{itX}]\\
    &= \mathbb{E}[\cos(tX) + i\sin(tX)]\\
    &= \mathbb{E}[\cos(tX)] + i\mathbb{E}[\sin(tX)]
\end{align}
\end{subequations}
And so we immediately have a relation that can be used to compute the first trigonometric moments:
\begin{subequations}\label{eq:multi_angle_trig_moments}
\begin{align}
    \mathbb{E}[\cos(tX)] = \text{Re}[\Phi_X(t)]\\
    \mathbb{E}[\sin(tX)] = \text{Im}[\Phi_X(t)]
\end{align}
\end{subequations}
The trigonometric power formulas can then be used to relate quantities of the form $\cos(mX)$ to $\cos^m(X)$ for any $m\in\N$ and similarly for the sin function. Applying the expectation operator to the trigonometric power formulas, applying the linearity of expectation and substituting in equation \ref{eq:multi_angle_trig_moments}, we have for some coefficients $s_0, c_0, s_{0,k}, s_{1,k}, c_{0, k}, s_{1, k}$ that $\forall n\in\N$ equations \ref{eq:higher_trig_moments} hold. The exact values of the coefficients can be found in \cite{trig_power}; they are omitted here for concision.
\begingroup
\allowdisplaybreaks
\begin{subequations}\label{eq:higher_trig_moments}
\begin{align}
    \mathbb{E}[\sin^{2n}(X)] &= s_0 + \sum_{k=0}^{n-1}s_{0,k}\text{Re}[\Phi_X(2(n-k))]\\
    \mathbb{E}[\sin^{2n+1}(X)] &= \sum_{k=0}^n s_{1,k}\text{Im}[\Phi_X(2n + 1 - 2k)]\\
    \mathbb{E}[\cos^{2n}(X)] &= c_{0} + \sum_{k=0}^{n-1}c_{0,k}\text{Re}[\Phi_X(2(n-k))]\\
    \mathbb{E}[\cos^{2n+1}(X)] &= \sum_{k=0}^nc_{1,k}\text{Re}[\Phi_X(2n + 1 - 2k)]
\end{align}
\end{subequations}
\endgroup
In the general case when $m, n > 0$, we can substitute in the exponential forms of $\sin$ and $\cos$ to get:
\begin{subequations}
\begin{align}
&\mathbb{E}[\cos^m(X)\sin^n(X)] = \\
    &\mathbb{E}\left[\frac{1}{i^n2^{m + n}}(e^{iX} + e^{-iX})^m(e^{iX} - e^{-iX})^n\right]\label{eq:cross_trig_moments}
\end{align}
\end{subequations}
Observe that the expression:
\begin{align}
    \frac{1}{i^n2^{m + n}}(e^{iX} + e^{-iX})^m(e^{iX} - e^{-iX})^n
\end{align}
can be written as a polynomial in $e^{iX}$ and $e^{-iX}$ by the closure property of $\mathbb{C}[X]$. An alternative view is that such an expression would be a \textit{Laurent polynomial} in $e^{iX}$ (Laurent polynomials are generalizations of polynomials that allow for negative powers). In either case, since the following equality holds $\forall t\in\Z$:
\begin{align}
    \mathbb{E}[e^{{iX}^t}] = \mathbb{E}[e^{itX}] = \Phi_{X}(t)
\end{align}
By expanding equation \ref{eq:cross_trig_moments} into a (possibly Laurent) polynomial and applying the linearity of expectation, we arrive at the weighted sum of $\Phi_X$ evaluated at multiple values. Like before, the (possibly Laurent) polynomial can be found using symbolic algebra.


\section{Example Application: DR-RRT with Stochastic Dubin's Dynamics}\label{sec:dr_rrt}
As an example application of this moment propagation framework, we present a version of distributionally robust RRT (DR-RRT) with uncertainty propagated through nonlinear Dubin's car dynamics. We make two major improvements to the original DR-RRT formulation: 1) we propagate moments through nonlinear dynamics as opposed to linear systems in the original formulation and 2) we use a different method to assess risk that is strictly less conservative \cite{dr_rrt}. Throughout this section, the robot is modelled with the stochastic Dubin's car described in equation \ref{eq:dubins_poly} so $\mathbf{x}_t = [x_t, y_t, v_t, c_t, s_t]$ denotes the state of the robot and $\mathbf{p}_t = [x_t, y_t]$ denotes its position. In this section, for any random vector $\mathbf{v}$, $\mu_{\mathbf{v}}, \Sigma_{\mathbf{v}}$ will denote its mean vector and covariance matrix. Algorithm \ref{alg:nonlinear_dr_rrt} describes this nonlinear DR-RRT algorithm; the key points of interest are the $\text{StochasticSteer}$ and $\text{RiskBound}$ functions which are described in the following subsections.
\subsection{Stochastic Steering}
It is common in RRTs, especially for ground vehicles, to perform steering between nodes with Dubins paths \cite{karaman2011anytime, dubins1957curves}. In our formulation, a deterministic Dubin's path is calculated, then control inputs for steering, $u_{\theta_t}$, to track the Dubin's path at a constant speed are computed with the discretized deterministic Dubin's dynamics. The steering disturbance in our stochastic system given by Eq \ref{eq:dubins_poly}, $w_{\theta_t},$ is then offset by $u_{\theta_t}$ to account for both the control and uncertainty in a single term. In an offline step, we run \textit{TreeRing} on the system in equation \ref{eq:dubins_poly} with an initial moment basis, $\mathcal{A}$, containing the first and second moments of $x$ and $y$:
\begin{align}
    \mathbf{x}[\mathcal{A}] = \{x, y, xy, x^2, y^2\}
\end{align}
And we arrive at the following completion w.r.t. the reduced MUF $\mathbf{x}\left[\bar{\mathcal{A}}\right]$:
\begin{align}
\begin{split}\label{eq:comp_moment_basis}
    \mathbf{x}[\mathcal{A}]\cup
    \{&c, s, v, v^2, xs, ys, xc, yc, s^2,\\
      & c^2, cs, xvs, xvc, yvs, yvc\}
\end{split}
\end{align}
We use this completion as the completion w.r.t the un-reduced MUF also has the moments:
\begin{align}
    \{s^2v^2, s^2v, csv^2, c^2v, c^2v^2\}
\end{align}
requiring more equations in the system. The reduced MUF for each $\alpha\in\mathcal{A}$ is then found using symbolic algebra to arrive at our moment state dynamical system. In this way, we can compute a sequence of mean vectors and covariance matrices to steer the vehicle from an existing node to a new one. To see the full set of equations associated with these moment bases, see the example provided in the Github repository.

\begin{figure}[!b]
    \centering
    \includegraphics[width = 0.9\linewidth]{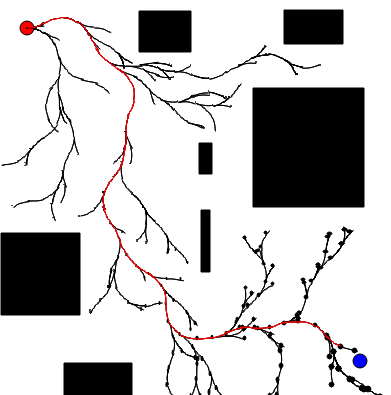}
    \caption{Example results from our DR-RRT. Non-Gaussian distributions in general are difficult to visualize as the concept of confidence ellipses don't apply. Thus, the sizes of the nodes correspond to the one $\sigma$ confidence ellipses if the position distribution is actually Gaussian. For this experiment, $w_{v_t}, w_{\theta_t}\sim\mathcal{N}(0, 10^{-8})$. A chance constraint of 0.1 was set for the whole open loop planning horizon.}
    \label{fig:my_label}
\end{figure}
\subsection{Bounding Risk}
In the environment, $n_o$ polytopic obstacles are represented by the intersection of half-spaces defined by a collection of linear inequalities:
\begin{align}
    \chi_{obs_i} = \cap_{j=1}^{n_i}\{\mathbf{x}\in\mathbb{R}^2 : a_{ij}^T\mathbf{x} + b_{ij}\leq 0\} \hquad i\in[n_o]
\end{align}
We are interested in ensuring that across a $T$ time step horizon, the probability of collision with an obstacle is less than some $\epsilon\in(0, 1)$. This risk can be upper bounded as such by applying Boole's Inequaltiy twice:
\begin{align}
    \mathbb{P}\bigg(\bigcup_{t=1}^T\bigcup_{i = 1}^{n_o}\{\mathbf{p}_t \in \chi_{obs_i}\}\bigg)\leq \sum_{t=1}^T\sum_{i=1}^{n_0}\mathbb{P}(\{\mathbf{p}_t\in\chi_{obs_i}\})\label{eq:booles}
\end{align}
So it is sufficient to establish upper bounds on $\mathbb{P}(\{\mathbf{p}_t\in\chi_{obs_i}\})$ for all times and for all obstacles. Since the probability of the intersection of a collection of events is less than the probability of any individual event occurring, for the $i_{th}$ obstacle, we have that:
\begin{align}
    \mathbb{P}(\mathbf{p}_t\in\chi_{obs_i}) \leq \min_{j\in[n_i]} \mathbb{P}(\{a_{ij}^T\mathbf{p}_t + b_{ij}\leq 0\})\label{eq:min_bound}
\end{align}
That is, the probability of the robot being in the $i_{th}$ obstacle is upper bounded by the probability of violating any one of the linear inequalities that define the obstacle; so we can upper-bound the probability of being in the $i_{th}$ obstacle with an upper-bound on the probability of being in the least risky half-space that defines it. Departing from previous works, we directly find an upper bound on risk as opposed to checking a sufficient condition (more discussion in \ref{subsec:comparison}) \cite{calafiore2006distributionally, dr_rrt}. By Cantelli's inequality, for any measurable function $g$, whenever $\mu_{g(\mathbf{p}_t)} \geq 0$:
\begin{equation}
    \Prob(g(\mathbf{p}_t)\leq 0)\leq \frac{\text{Var}[g(\mathbf{p}_t)]}{\text{Var}[g(\mathbf{p}_t)] + \mathbb{E}[g(\mathbf{p}_t)]^2}\label{eq:cantelli}
\end{equation}
Note that the restriction that $\mu_{g(\mathbf{p}_t)}\geq 0$ is not particularly restrictive as it is only not satisfied when the average case corresponds to collision. Thus, we simply let the risk bound be $1$ whenever $\mu_{g(\mathbf{p}_t)}< 0$. In our case, we have that $g(\mathbf{p}_t) = a_{ij}^T\mathbf{p}_t + b_{ij}$ for each linear inequality. The mean and covariance can be computed in terms of $\mu_{\mathbf{p}_t}$ and $\Sigma_{\mathbf{p}_t}$ with:
\begin{subequations}
\begin{align}
    &\mathbb{E}[a_{ij}^T\mathbf{p}_t + b_{ij}] = a_{ij}^T\mu_{\mathbf{p}_t} + b_{ij}\\
    &\text{Var}[a_{ij}^T\mathbf{p}_t + b_{ij}] = a_{ij}^T\Sigma_{\mathbf{p}_t}a_{ij}
\end{align}
\end{subequations}
Note that since the Cantelli bound holds for any measurable function $g$, in principle, we can extend this formulation to the case of non-convex obstacles described as the sub-level sets of polynomials. However, since we need the mean and variance of $g(\mathbf{p}_t)$, doing so would require higher moments of $\mathbf{p}_t$, so more moments need to be propagated.

\begin{algorithm}
\caption{Nonlinear DR-RRT}\label{alg:nonlinear_dr_rrt}
\textbf{Input} :  RRT tree $\mathcal{T}$, obstacles $\chi_{obs}$,  disturbance model $\mathbf{w}$
\begin{algorithmic}[1]
\Procedure{BuildRRT}{$\mathcal{T}, \chi_{obs}, \mathbf{w}$}
\State $(n_{new}, n_{near})\leftarrow \text{Sample and find the nearest node}$
\State $\mathbf{u}_{0:T}\leftarrow \text{Steer}(n_{new}, n_{near})$
\State $(\mu_{\mathbf{x}_{0:T}}, \Sigma_{\mathbf{x}_{0:T}} )\leftarrow \text{StochasticSteer}(n_{near}, \mathbf{u}_{0:T}, \mathbf{w})$
\State $\epsilon_{risk}\leftarrow \text{RiskBound}(\mu_{\mathbf{x}_{0:T}}, \Sigma_{\mathbf{x}_{0:T}}, \chi_{obs})$
\State $\epsilon_{new} = \epsilon_{risk} + \text{RiskToNode}(n_{near})$
\If{$\epsilon_{new} < \epsilon$}
    \State $\text{RiskToNode}(n_{new}) = \epsilon_{new}$
    \State Add $\mu_{\mathbf{x}_T}$ and $\Sigma_{\mathbf{x}_T}$ to $n_{new}$
    \State Add $n_{new}$ and $(n_{near}, n_{new})$ to $\mathcal{T}$
\EndIf\Comment{Rewire for RRT*}
\EndProcedure
\end{algorithmic}
\end{algorithm}

\subsection{Comparison with Prior DR Constraints}\label{subsec:comparison}
In prior works \cite{calafiore2006distributionally, dr_rrt}, given a chance-constraint of $\epsilon$ across the whole horizon, a maximum allowable risk for each time step and each obstacle is assigned s.t. $\sum_{t=1}^T\sum_{i=1}^{n_o}\epsilon_{t,i} = \epsilon$. The following sufficient condition for the $\epsilon_{t,i}$ chance-constraint to be satisfied is then checked for each linear inequality:
\begin{align}
    \mathbb{E}[a_{ij}^T\mathbf{p}_t + b_{ij}]\geq \sqrt{a_{ij}^T\Sigma_{\mathbf{p}_t}a_{ij}}\sqrt{\frac{1-\epsilon_{t,i}}{\epsilon_{t,i}}}\label{eq:distr_robust}
\end{align}
And collision is defined as the above condition not being satisfied. We would like to note that this sufficient condition is ultimately equivalent to the bound we employ as simple manipulation of the equation results in:
\begin{subequations}
\begin{align}
\epsilon_{t,i}&\geq \frac{1}{1 + \frac{\mathbb{E}[a_{ij}^T\mathbf{p}_t + b_{ij}]^2}{a_{ij}^T\Sigma_{\mathbf{p}_t}a_{ij}}}\\
&= \frac{a_{ij}^T\Sigma_{\mathbf{p}_t}a_{ij}}{a_{ij}^T\Sigma_{\mathbf{p}_t}a_{ij} + \mathbb{E}[a_{ij}^T\mathbf{p}_t + b_{ij}]^2}\\
&= \frac{\text{Var}(a_{ij}^T\mathbf{p}_t + b_{ij})}{\text{Var}(a_{ij}^T\mathbf{p}_t + b_{ij}) + \mathbb{E}[a_{ij}^T\mathbf{p}_t + b_{ij}]^2}
\end{align}
\end{subequations}
Thus, our approach uses the same fundamental inequality. However, instead of allocating $\epsilon_{t,i}$ \textit{a priori}, we determine upper bounds on risk and ensure that the summation is less than the overall chance-constraint $\epsilon$ thus allowing for less conservative results.

\section{Comparison with Monte Carlo and Linearization}\label{sec:comparison}
This section uses the Dubin's car model to compare the proposed method for exact nonlinear moment propagation against naive Monte Carlo and moment propagation with a linearized system. We linearize the Dubin's car model about the initial state and, since we apply Euler integration, the matrices $A, B$ and vector $\mathbf{c}$ are scaled by the time step to arrive at the following discrete time system:
\begin{align}
    \mathbf{x}_{t+1} = (I + A)\mathbf{x}_t + B\mathbf{w}_t + \mathbf{c}
\end{align}
Where $I$ denotes the identity matrix. The mean vector and covariance matrix thus have the following dynamics:
\begin{subequations}
\begin{align}
    \mu_{\mathbf{x}_{t+1}} &= (I + A)\mu_{\mathbf{x}_t} + B\mu_{\mathbf{w}_t} + \mathbf{c}\\
    \Sigma_{\mathbf{x}_{t+1}} &= (I+A)\Sigma_{\mathbf{x}_t}(I+A)^T + B\Sigma_{\mathbf{w}_t}B^T
\end{align}
\end{subequations}
The distributions of the disturbances are chosen to be $w_{v_t}\sim \text{Beta}(10, 1000)$ and $w_{\theta_t}\sim\mathcal{N}(0.04, 0.03)$. Fig. \ref{fig:compare_moment_prop} compares the results of our nonlinear moment propagation method against naive Monte Carlo and the standard method of propagating a linearized system. We note that the Monte Carlo results closely match our moment propagation method, while the second moments of the linearized system diverge quite significantly from those obtained from the other two methods. The nonlinear moment propagation method is also very fast; in our C++ implementation, propagation over $10^5$ time steps takes only 35.2 ms for a computation time of ~$0.3$ microseconds per time step. This is on par with speeds for well-implemented moment propagation methods for linear systems and many orders of magnitude faster than naive Monte Carlo.
\begin{figure}
    \centering
    \includegraphics[width=\linewidth]{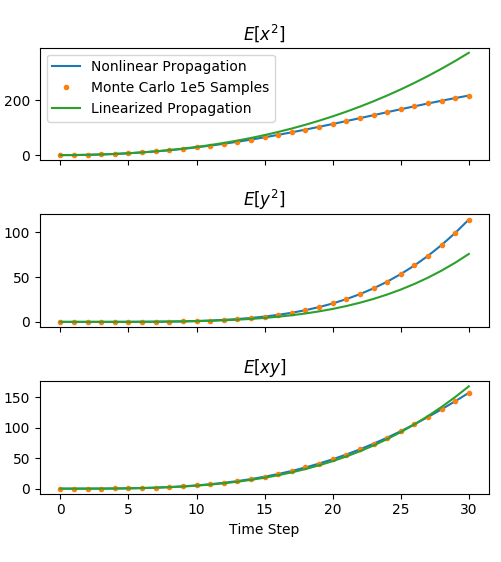}
    \caption{Comparison of propagated moments with our nonlinear moment propagation technique, naive Monte Carlo, and moment propagation with the linearized system.}
    \label{fig:compare_moment_prop}
\end{figure}

\section{Conclusions}
In this paper, we present a new framework for nonlinear moment propagation with wide applicability to problems in robotics and control, namely to trigonometric polynomial systems. The key result is an algorithm that can search for a moment state dynamical system that can be used to compute moments of the state vector given the moments of external disturbances. We demonstrate its application to a Dubin's car model with actuation uncertainty which is then used in an improved distributionally robust RRT. Finally, numerical experiments show that our method is as fast as moment propagation techniques that employ linearization while being exact. As \textit{TreeRing} can produce relatively large systems of equations, making manual conversion to code both tedious and prone to error, future work should consider developing code generation capabilities. There are many other potential applications of this exact moment propagation technique including but not limited to Gaussian process motion planning, trajectory optimization, and nonlinear estimation which should be investigated further in future works.
\bibliographystyle{plain}
\bibliography{references}

\begin{thebibliography}{10}

\bibitem{aoude2013probabilistically}
Georges~S Aoude, Brandon~D Luders, Joshua~M Joseph, Nicholas Roy, and
  Jonathan~P How.
\newblock Probabilistically safe motion planning to avoid dynamic obstacles
  with uncertain motion patterns.
\newblock {\em Autonomous Robots}, 35(1):51--76, 2013.

\bibitem{bertsimas2005optimal}
Dimitris Bertsimas and Ioana Popescu.
\newblock Optimal inequalities in probability theory: A convex optimization
  approach.
\newblock {\em SIAM Journal on Optimization}, 15(3):780--804, 2005.

\bibitem{Larss1}
Lars Blackmore, Hui Li, and Brian Williams.
\newblock A probabilistic approach to optimal robust path planning with
  obstacles.
\newblock {\em Proceedings of the American Control Conference Minneapolis,
  Minnesota, USA}, 2006.

\bibitem{Larss2}
Lars Blackmore and Masahiro Ono.
\newblock Convex chance constrained predictive control without sampling.
\newblock {\em AIAA Guidance, Navigation, and Control Conference Chicago,
  Illinois}, 2009.

\bibitem{calafiore2006scenario}
Giuseppe~C Calafiore and Marco~C Campi.
\newblock The scenario approach to robust control design.
\newblock {\em IEEE Transactions on automatic control}, 51(5):742--753, 2006.

\bibitem{calafiore2006distributionally}
Giuseppe~Carlo Calafiore and Laurent El~Ghaoui.
\newblock On distributionally robust chance-constrained linear programs.
\newblock {\em Journal of Optimization Theory and Applications}, 130(1):1--22,
  2006.

\bibitem{carron2019data}
Andrea Carron, Elena Arcari, Martin Wermelinger, Lukas Hewing, Marco Hutter,
  and Melanie~N Zeilinger.
\newblock Data-driven model predictive control for trajectory tracking with a
  robotic arm.
\newblock {\em IEEE Robotics and Automation Letters}, 4(4):3758--3765, 2019.

\bibitem{dubins1957curves}
Lester~E Dubins.
\newblock On curves of minimal length with a constraint on average curvature,
  and with prescribed initial and terminal positions and tangents.
\newblock {\em American Journal of mathematics}, 79(3):497--516, 1957.

\bibitem{fisher2011optimal}
James Fisher and Raktim Bhattacharya.
\newblock Optimal trajectory generation with probabilistic system uncertainty
  using polynomial chaos.
\newblock {\em Journal of dynamic systems, measurement, and control}, 133(1),
  2011.

\bibitem{hewing2019cautious}
Lukas Hewing, Juraj Kabzan, and Melanie~N Zeilinger.
\newblock Cautious model predictive control using gaussian process regression.
\newblock {\em IEEE Transactions on Control Systems Technology}, 2019.

\bibitem{jasour2018moment}
Ashkan~M Jasour, Andreas Hofmann, and Brian~C Williams.
\newblock Moment-sum-of-squares approach for fast risk estimation in uncertain
  environments.
\newblock In {\em 2018 IEEE Conference on Decision and Control (CDC)}, pages
  2445--2451. IEEE, 2018.

\bibitem{jasour2019risk}
Ashkan~M Jasour and Brian~C Williams.
\newblock Risk contours map for risk bounded motion planning under perception
  uncertainties.
\newblock {\em Robot.; Sci. and Syst.}, 2019.

\bibitem{karaman2011anytime}
Sertac Karaman, Matthew~R Walter, Alejandro Perez, Emilio Frazzoli, and Seth
  Teller.
\newblock Anytime motion planning using the rrt.
\newblock In {\em 2011 IEEE International Conference on Robotics and
  Automation}, pages 1478--1483. IEEE, 2011.

\bibitem{kim2012generalized}
Kwang-Ki~K Kim and Richard~D Braatz.
\newblock Generalized polynomial chaos expansion approaches to approximate
  stochastic receding horizon control with applications to probabilistic
  collision checking and avoidance.
\newblock In {\em 2012 IEEE International Conference on Control Applications},
  pages 350--355. IEEE, 2012.

\bibitem{kloeden2013numerical}
Peter~E Kloeden and Eckhard Platen.
\newblock {\em Numerical solution of stochastic differential equations},
  volume~23.
\newblock Springer Science \& Business Media, 2013.

\bibitem{CCRRT}
Brandon~D. Luders, Mangal Kothari, and Jonathan~P. How.
\newblock Chance constrained rrt for probabilistic robustness to environmental
  uncertainty.
\newblock {\em Proceedings of the AIAA Guidance, Navigation, and Control
  Conference}, 2010.

\bibitem{melchior2007particle}
Nik~A Melchior and Reid Simmons.
\newblock Particle rrt for path planning with uncertainty.
\newblock In {\em Proceedings 2007 IEEE International Conference on Robotics
  and Automation}, pages 1617--1624. IEEE, 2007.

\bibitem{meurer2017sympy}
Aaron Meurer, Christopher~P Smith, Mateusz Paprocki, Ond{\v{r}}ej
  {\v{C}}ert{\'\i}k, Sergey~B Kirpichev, Matthew Rocklin, AMiT Kumar, Sergiu
  Ivanov, Jason~K Moore, Sartaj Singh, et~al.
\newblock Sympy: symbolic computing in python.
\newblock {\em PeerJ Computer Science}, 3:e103, 2017.

\bibitem{nakka2019trajectory}
Yashwanth~Kumar Nakka and Soon-Jo Chung.
\newblock Trajectory optimization for chance-constrained nonlinear stochastic
  systems.
\newblock 2019.

\bibitem{dr_rrt}
Tyler Summers.
\newblock Distributionally robust sampling-based motion planning under
  uncertainty.
\newblock {\em IEEE/RSJ International Conference on Intelligent Robots and
  Systems (IROS)}, 2018.

\bibitem{LQG1}
Jur van~den Berg1, Pieter Abbeel, and Ken Goldberg.
\newblock Lqgmp: Optimized path planning for robots with motion uncertainty and
  imperfect state information.
\newblock {\em The International Journal of Robotics Research Vol. 30(7),
  895--913, 2011.}, 2011.

\bibitem{wang2020non}
Allen Wang, Ashkan Jasour, and Brian Williams.
\newblock Non-gaussian chance-constrained trajectory planning for autonomous
  vehicles in the presence of uncertain agents.
\newblock {\em arXiv preprint arXiv:2002.10999}, 2020.

\bibitem{trig_power}
Eric~W. Weisstein.
\newblock Trigonometric power formulas. {From MathWorld---A Wolfram Web
  Resource}.

\bibitem{xiu2010numerical}
Dongbin Xiu.
\newblock {\em Numerical methods for stochastic computations: a spectral method
  approach}.
\newblock Princeton university press, 2010.

\end{thebibliography}
\end{document}